\documentclass[10pt]{amsart}  
\usepackage{amsfonts}
\usepackage{amssymb}
\usepackage{amsmath,amsthm}
\usepackage{amscd}
\usepackage{enumitem}
\usepackage{cancel}
\usepackage{pdfsync}
\usepackage{mathrsfs}
\usepackage{stmaryrd}
\usepackage[colorlinks,linkcolor=blue,citecolor=blue,urlcolor=blue]{hyperref}

\newcommand{\rd}{\mathrm{d}}

\newcommand{\beq}{\begin{equation}}
\newcommand{\eeq}{\end{equation}}
\newcommand{\bea}{\begin{eqnarray}}
\newcommand{\eea}{\end{eqnarray}}
\newcommand{\nn}{\nonumber}
\newcommand\noi{\noindent}

\newcommand{\bs}{\boldsymbol}
\newcommand{\bk}{\begin{cases}}
\newcommand{\ek}{\end{cases}}
\newcommand{\tbf}{\textbf}

\newtheorem{theorem}{Theorem}[section]
\newtheorem{lemma}[theorem]{Lemma}

\theoremstyle{definition}
\newtheorem{definition}[theorem]{Definition}
\newtheorem{example}[theorem]{Example}

\theoremstyle{remark}
\newtheorem{remark}[theorem]{Remark}

\newtheorem{proposition}[theorem]{Proposition}

\numberwithin{equation}{section}

\begin{document}

\title[Polarization of generalized Nijenhuis torsions]{Polarization of generalized Nijenhuis torsions}

\author{Piergiulio Tempesta}
\address{Departamento de F\'{\i}sica Te\'{o}rica, Facultad de Ciencias F\'{\i}sicas, Universidad
Complutense de Madrid, 28040 -- Madrid, Spain \\  and Instituto de Ciencias Matem\'aticas, C/ Nicol\'as Cabrera, No 13--15, 28049 Madrid, Spain}
\email{piergiulio.tempesta@icmat.es, ptempest@ucm.es}
\thanks{The research of P. T. has been supported by the research project
PGC2018-094898-B-I00, Ministerio de Ciencia, Innovaci\'{o}n y Universidades and Agencia Estatal de Investigaci\'on,
Spain, and by the Severo Ochoa Programme for Centres of Excellence in R\&D
(CEX2019-000904-S), Ministerio de Ciencia, Innovaci\'{o}n y Universidades y Agencia Estatal de Investigaci\'on, Spain. 
P. T. is member of the Gruppo Nazionale di Fisica Matematica (GNFM)}
\author{Giorgio Tondo}
\address{Dipartimento di Matematica e Geoscienze, Universit\`a  degli Studi di Trieste,
piaz.le Europa 1, I--34127 Trieste, Italy.}
\email{tondo@units.it}
\thanks{G. Tondo thanks prof. V. Beorchia for useful discussions. The research of G. T. has been supported by the research project FRA2021-2022, Universit\'a degli Studi di Trieste, Italy. 
}
\dedicatory{In memory of Professor Alexandre Vinogradov}

\date{February 5, 2022}

\begin{abstract}
%The theory of generalized Nijenhuis torsions, recently introduced, offers new powerful tools to detect the Frobenius integrability of operator fields on a differentiable manifold. 

%They consist of a $C^{\infty}$-modulus of Haantjes operators, closed under the product operation.

In this work, we introduce the notion of polarization  of  generalized  Nijenhuis torsions and establish several algebraic identities. We prove that these polarizations are relevant in the characterization of Haantjes $C^{\infty}$(M)-modules of operator fields.

%However, in general  the construction of algebras and, even, moduli of Haantjes operators over the ring of smooth functions on a differentiable manifold is not an easy task.
%In this paper we show how to define in a new, systematic way  a class of vector-valued two forms which generalize the Frölicher-Nijenhuis bracket. To this aim, we shall use a  family of generalized Nijenhuis torsions recently introduced by us and a suitable polarization identity. In particular, the first example of our construction will be relevant in the characterization of Haantjes moduli of operators.
\end{abstract}

\maketitle

\tableofcontents

\section{Introduction}

In the last years, the theory of Haantjes operators, namely (1,1)-tensor fields having a vanishing Haantjes torsion \cite{Haa1955},  has experienced a resurgence of interest.  New applications have been found, for instance in the theory of infinite-dimensional integrable systems \cite{MGall13}-\cite{MGall17} and of hydrodynamic-type systems \cite{FeKhu},\cite{FM2007MA}. Recently, Haantjes operators have also been related with Hamiltonian integrable and superintegrable models and with the classical problem of separation of variables in Hamiltonian mechanics  \cite{TT2016}, \cite{TT2021AMPA}, \cite{RTT2022CNS}.

In \cite{TT2021JGP} the notion of Haantjes algebra has been introduced. Its relevance is due to the fact that, in the semisimple, Abelian case, one can construct a coordinate chart in which all of the operators of the algebra can be diagonalized simultaneously.

Our theory of generalized Nijenhuis torsions of level $m$ has been introduced in \cite{TT2021JGP} and further studied in \cite{TT2022CMP}. In particular, it was proved that the vanishing of a generalized Nijenhuis torsion of a given operator field on a manifold is sufficient to guarantee that each eigendistribution as well as each direct sum of them is Frobenius integrable (mutual integrability). This result generalizes the standard integrability theorems of Nijenhuis \cite{Nij1955} and Haantjes \cite{Haa1955} to a much wider class of operators. In addition, we proved that coordinate charts exist allowing us to represent the operator in a block-diagonal form.

%In \cite{RTT2022arx}, generalized Haantjes algebras of level $m$ have been introduced, namely  algebras of operators whose generalized Nijenhuis torsion of a given level $m$ vanishes. One can ensure, under some technical conditions,  the simultaneous block-diagonalization (in suitable coordinate charts) of the operators forming a generalized Haantjes algebra.

The aim of this work is to introduce the notion of \textit{polarization of a generalized Nijenhuis torsion of level $m$}. In the spirit of the theory of polarization of homogeneous forms, our construction allows us to define a new vector-valued 2-form, depending on a set of operators $\{\bs{A}_1,\ldots,\bs{A}_{2m}\}$, $m\geq 1$. For $m=1$, it coincides with the standard Fr\"olicher-Nijenhuis bracket. 

A direct algebraic application of the polarization of generalized Nijenhuis torsions concerns  the theory of Haantjes modules, namely $C^{\infty}$-modules of operators having vanishing Haantjes torsion.  Our main result is Theorem \ref{th:main}, stating that the vanishing of the polarization of a generalized Nijenhuis torsion of level $m$, evaluated on suitable arguments, is necessary and sufficient for the existence of an Abelian generalized Haantjes module of the same level.

%More important, the theory of polarization forms that we establish in this work offers a natural and unified language allowing us to interpret the theory of generalized Nijenhuis torsions from a new perspective. In turn, this general framework enables us to obtain new results interesting per se from an algebraic viewpoint. 

\section{Generalized Nijenhuis torsions}
\label{sec:1}
In this section, we shall review our recent construction of generalized Nijenhuis torsions \cite{TT2021JGP} and Haantjes brackets \cite{TT2022CMP}. 

Let $M$  be a differentiable manifold. We shall denote by $\mathfrak{X}(M)$ the Lie algebra of smooth vector fields on $M$   and  by $\mathcal{T}^{1}_{1}(M)$ the space of smooth $(1,1)$-tensor fields on $M$ (operator fields). For the sake of simplicity, from now on the expressions ``tensor fields'' and ``operator fields'' will be abbreviated to tensors and operators. In the following, all tensors will be assumed to be smooth. 
\subsection{Generalized Nijenhuis and Haantjes geometries}

\begin{definition} [\cite{Nij1951}]\label{def:N}
Let $\boldsymbol{A}:\mathfrak{X}(M)\rightarrow \mathfrak{X}(M)$ be an operator in  $\mathcal{T}^{1}_{1}(M)$. The
 \textit{Nijenhuis torsion} of $\boldsymbol{A}$ is the vector-valued $2$-form  defined by
\begin{equation} \label{eq:Ntorsion}
\tau_ {\boldsymbol{A}} (X,Y):=\boldsymbol{A}^2[X,Y] +[\boldsymbol{A}X,\boldsymbol{A}Y]-\boldsymbol{A}\Big([X,\boldsymbol{A}Y]+[\boldsymbol{A}X,Y]\Big),
\end{equation}
where $X,Y \in \mathfrak{X}(M)$ and $[ \ , \ ]$ denotes the Lie bracket of  vector fields.
\end{definition}
We recall the definition of generalized torsions,  according to the formulation of \cite{TT2021JGP}. 
\begin{definition} \label{df:mtorsion}
Let $\boldsymbol{A}:\mathfrak{X}(M)\to \mathfrak{X}(M)$ be an operator. For each integer $m\geq 1$, the generalized Nijenhuis torsion of $\boldsymbol{A}$ of level $m$  is the vector--valued $2$--form defined recursively by
\bea \label{GNTn}
\nn \mathcal{\tau}^{(m)}_{\boldsymbol{A}}(X,Y):&=&
\boldsymbol{A}^2\mathcal{\tau}^{(m-1)}_{\boldsymbol{A}}(X,Y)+
\mathcal{\tau}^{(m-1)}_{\boldsymbol{A}}(\boldsymbol{A}X,\boldsymbol{A}Y) \\ &-& 
\boldsymbol{A}\Big(\mathcal{\tau}^{(m-1)}_{\boldsymbol{A}}(X,\boldsymbol{A}Y)+\mathcal{\tau}^{(m-1)}_{\boldsymbol{A}}(\boldsymbol{A}X,Y)\Big), \quad X,Y \in \mathfrak{X}(M) \ .
\eea
Here the notations $\tau_{\boldsymbol{A}}^{(0)}(X,Y):= [X,Y]$,  $\tau_{\boldsymbol{A}}^{(1)}(X,Y):=\tau_{\boldsymbol{A}}(X,Y)$
%and $\tau_{\boldsymbol{A}}^{(2)}(X,Y):=\mathcal{H}_{\boldsymbol{A}}(X,Y)$ 
are used.
\end{definition}
\noi We recall a useful formula, proved in \cite{KS2017} (Section 4.6):
\begin{equation} \label{eq:K17}
\mathcal{\tau}^{(m)}_{{\boldsymbol{A}}}(X,Y)=\sum_{p=0}^m \sum_{q=0}^{m} (-1)^{2m-p-q} \binom{m}{p} \binom{m}{q} \bs{A}^{p+q}\big[ \bs{A}^{m-p}X,\bs{A}^{m-q} Y\big ] \ .
\end{equation}
This formula can also be proved by induction over $m$.

\begin{example} \label{def:H}
 \noi For $m=2$ one finds that $\mathcal{\tau}^{(m)}_{{\boldsymbol{A}}}(X,Y)$ coincides with  the \textit{Haantjes torsion} of $\boldsymbol{A}$, that is,   the vector-valued $2$-form defined by \cite{Haa1955}
\begin{equation} \label{eq:Haan}
\mathcal{H}_{\boldsymbol{A}}(X,Y):=\boldsymbol{A}^2\tau_{\boldsymbol{A}}(X,Y)+\tau_{\boldsymbol{A}}(\boldsymbol{A}X,\boldsymbol{A}Y)-\boldsymbol{A}\Big(\tau_{\boldsymbol{A}}(X,\boldsymbol{A}Y)+\tau_{\boldsymbol{A}}(\boldsymbol{A}X,Y)\Big).
\end{equation}
\end{example}

\begin{definition}[\cite{TT2022CMP}]
A generalized Nijenhuis  operator of level $m$ is a (1,1)-tensor  whose  generalized Nijenhuis torsion of the same level vanishes  identically.
\end{definition}
In particular, for $m=1$ and $m=2$, we get the standard Nijenhuis and Haantjes operators.
\noi The relevance of Haantjes operators in the theory of integrable Hamiltonian classical systems and for the construction of separating variables has been recently discussed in \cite{TT2021AMPA}, \cite{RTT2022CNS},  \cite{T2017}.

Another useful geometric object is the vector-valued $(1,2)$-tensor 
$\bs{\sf{T}}(\bs{A},\bs{B}): \mathfrak{X}^*(M) \times \mathfrak{X}(M)\times \mathfrak{X}(M)\rightarrow \mathfrak{X}(M)$ 
defined, for each pair of operators $\bs{A},\bs{B} \in \mathcal{T}^{1}_{1}(M)$, by  
\begin{equation}\label{eq:T}
\bs{\sf{T}}(\bs{A},\bs{B})(\alpha,X,Y):=(\bs{I}\otimes \bs{AB} -\bs{A}\otimes \bs{B})(\alpha,X,Y) \ .
\end{equation}
We shall denote by $\bs{\sf{T}}^T(\alpha,X,Y):=\bs{\sf{T}}(\alpha,Y,X)$ the transpose   of $\bs{\sf{T}}$ with respect to the last two arguments. We also recall that for each operator $\bs{A}$, $\bs{B}$, and for all $\alpha\in \mathfrak{X}^*(M)$, $X,Y \in \mathfrak{X}(M)$,
$$
(\bs{A}\otimes \bs{B}) (\alpha,X,Y)=\langle \alpha, \bs{A}X\rangle \,\bs{B}Y \ .
$$
Thus, one can prove the following
\begin{proposition} [\cite{TT2022CMP}]\label{corollary:3}
Let $\bs{A}:\mathfrak{X}(M)\to \mathfrak{X}(M)$ be an operator. Then, for all $f, g\in C^{\infty}(M)$ the relations 
\begin{eqnarray} \label{tauind}
\tau_{f\bs{A}}(X,Y)&=&f^2 \tau_{\bs{A}}(X,Y)- f\big(\bs{\sf{T}}(\bs{A},\bs{A})-\bs{\sf{T}}^T(\bs{A},\bs{A})\big)(\rd f,X,Y)
\\  
\label{tauind2}
\mathcal{\tau}^{(m)}_{f \bs{A} + g\bs{I} }(X,Y)&=& f^{2m}\mathcal{\tau}^{(m)}_{\boldsymbol{A}}(X,Y),\quad\qquad \qquad m\in\mathbb{N}\backslash \{0,1\} 
\end{eqnarray}
hold.
\end{proposition}

\par
According to eq. \eqref{tauind2}, we deduce that $\mathcal{\tau}^{(m)}_{\bs{A}}(X,Y)$ is a homogeneous 2-form of degree $2m$ in $\bs{A}$.

\section{Polarization of generalized Nijenhuis  torsions}
In this section, we shall introduce the polarization form of a generalized Nijenhuis torsion. To this aim, we shall start with the simplest representative of our construction.
\subsection{The Fr\"olicher-Nijenhuis bracket}
\begin{definition}[\cite{FN1956}]
Let $\bs{A}$, $\bs{B}\in \mathcal{T}^{1}_{1}(M)$.
 The \textit{Fr\"olicher--Nijenhuis bracket} of $\boldsymbol{A}$ and $\boldsymbol{B}$ is the vector-valued $2$-form 
%the skew-symmetric  $(1,2)$ tensor field 
given by \footnote{For the sake of clarity, in this article we have renounced the usual unified notation $[ \cdot \ , \cdot ]$ which, depending on the context, should stand for both  the standard Lie bracket of vector fields and  the Fr\"olicher--Nijenhuis bracket
of operators. Instead, we have preferred to maintain the symbol $\llbracket\cdot \ ,\cdot\rrbracket$ for the Fr\"olicher--Nijenhuis bracket  and to introduce  the notation $[ \cdot \ , \cdot ]$  for the Lie bracket of two vector fields as well as the commutator of two operators.}
\bea \label{eq:binNtors}
%\nn
& & \llbracket\bs{A},\bs{B}\rrbracket(X,Y):= 
 %\mathcal{T}_ {\boldsymbol{A}+\boldsymbol{B}} (X,Y) - \big( \mathcal{T}_ {\boldsymbol{A}} (X,Y)  + \mathcal{T}_ {\boldsymbol{B}} (X,Y)\big)\\
 \Big(\bs{AB}+\bs{BA}\Big)[X,Y]+[\bs{A}X,\bs{B}Y]+[\bs{B}X,\bs{A}Y] \\
 \nn 
 &&-\bs{A}\Big([X,\bs{B}Y]+[\bs{B}X,Y]\Big )
-\bs{B}\Big ([X,\bs{A}Y]+[\bs{A}X,Y]\Big), \qquad X,Y \in \mathfrak{X}(M) \ .
\eea
\end{definition}
We observe that the Fr\"olicher--Nijenhuis bracket can be regarded as the polarization of the Nijenhuis torsion.
Indeed, as the Nijenhuis torsion \eqref{eq:Ntorsion} is a \textit{quadratic form} involving the entries of the operator $\bs{A}$, it can be naturally polarized. As a result, one obtains the associated bilinear form that coincides with the Fr\"olicher-Nijenhuis bracket:
\begin{equation} \label{FNB}
 \llbracket\bs{A},\bs{B}\rrbracket(X,Y)=  \tau_ {\boldsymbol{A}+\boldsymbol{B}} (X,Y) - \big( \tau_ {\boldsymbol{A}} (X,Y)  + \tau_ {\boldsymbol{B}} (X,Y)\big) \ .
\end{equation}
%the polarization (up to a factor $1/2$) of the Nijenhuis torsion

\noi We remind that the local expression of the components of the Fr\"olicher-Nijenhuis bracket reads
\begin{equation}\label{eq:FNlocal}
\llbracket\bs{A},\bs{B}\rrbracket^i_{jk}=\sum_{l=1}^n \bigg(\bs{A}^l_{[j}\partial_{|l|}\bs{B}^i_{k]}-\bs{A}^i_l\partial_{[j}\bs{B}^l_{k]} +\bs{B}^l_{[j}\partial_{|l|}\bs{A}^i_{k]}-\bs{B}^i_l\partial_{[j}\bs{A}^l_{k]}
\bigg) \ .
\end{equation}
This bracket has relevant geometric applications \cite{NN1957}, in particular in the theory of almost-complex structures and in the detection of obstructions to integrability \cite{KMS1993}.
The  bracket is symmetric and $\mathbb{R}$-linear (but not $C^\infty(M)$-linear)  in  $\boldsymbol{A}$ and $\boldsymbol{B}$.  
In fact, it satisfies the identity
\begin{equation} \label{eq:FNf}
\llbracket f \bs{A}, \bs{B}\rrbracket(X,Y)
=f  \llbracket\bs{A},\bs{B}\rrbracket(X,Y)
- \big(\bs{\sf{T}}(\bs{B},\bs{A})-\bs{\sf{T}}^T(\bs{B},\bs{A})(\rd f,X,Y)\big) 
\end{equation}

\noi Choosing $\bs{B}=\bs{A}$ in Eq. \eqref{eq:binNtors}, one gets twice the Nijenhuis torsion:
\[
\llbracket \boldsymbol{A,A}\rrbracket (X,Y) = 2~\tau_ {\boldsymbol{A}} (X,Y)  \ .
\]

\subsection{Defect and polarization of torsions}
In order  to generalize equation \eqref{FNB}  to the case of Nijenhuis torsions of level $m$,
 we will follow the approach proposed in \cite{Drap}. Precisely, we first introduce  the notion of \textit{defect of a torsion of level $m$}, and consequently the polarization 2-form associated with it.
\begin{definition}
Let $\bs{A}_1,\ldots,\bs{A}_n \in \mathcal{T}^{1}_{1}(M)$. The defect  of index $k$ of the generalized Nijenhuis torsion of level $m$ is the vector-valued 2--form defined by
\begin{equation} \label{eq:Delta}
\Delta^{(m)}_k (\bs{A}_1,\bs{A}_2,\ldots,\bs{A}_{k})(X,Y):= 
%\frac{1}{m!} 
\sum_{\mathcal{I} \subseteq \{1,2,\ldots,k\}} (-1)^{k-|\mathcal{I}| }\ \mathcal{\tau}^{(m)}_{{\sum_{i\in \mathcal{I}} \bs{A}_{i}}}(X,Y) \ .
\end{equation}
Here $k,m\in \mathbb{N}\setminus \{0\}$, $\mathcal{I}$ denotes any non-empty subset extracted from $\{1,\ldots,k \}$ and $|\mathcal{I}|$ is the cardinality of $\mathcal{I}$.
\end{definition}
The  defect of index $k$  is by construction symmetric in its arguments $(\bs{A}_1, \ldots, \bs{A}_k)$.
It easy to prove from Eq. \eqref{eq:Delta} that the defect fulfils the recurrence relations 
\begin{eqnarray} \label{eq:DeltaRec}
&\quad \Delta^{(m)}_{k+1} (\bs{A}_1,\bs{A}_2,\ldots,\bs{A}_{k},\bs{A}_{k+1})(X,Y)= \Delta^{(m)}_{k} (\bs{A}_1+\bs{A}_2,\ldots,\bs{A}_{k},\bs{A}_{k+1})(X,Y) \\
\nn&-\Delta^{(m)}_{k} (\bs{A}_1,\bs{A}_3,\ldots,\bs{A}_{k},\bs{A}_{k+1})(X,Y)-\Delta^{(m)}_{k} (\bs{A}_2,\bs{A}_3,\ldots,\bs{A}_{k},\bs{A}_{k+1})(X,Y) \end{eqnarray}
involving the first argument. Analogous relations hold for each of the remaining arguments.

As a direct consequence of the recurrence relations \eqref{eq:DeltaRec} and their analogous ones, we can state the following
\begin{proposition}\label{pr:Deltanull}
  The defect  of index $k$  is additive in each of the entries $\bs{A}_i$ if and only if it satisfies the relation
 \begin{equation}
  \Delta^{(m)}_{k+1} (\bs{A}_1,\bs{A}_2,\ldots,\bs{A}_{k+1})(X,Y)= \bs{0} \ .
 \end{equation}
  \end{proposition}
According to the polarization identity stated in \cite{Land} (pag. 42), by specializing the  defect of index $k$ to the case $k=2m$, we are able to define  the new vector-valued 2-form representing the polarization form associated with our generalized Nijenhuis torsions.
\begin{definition}
 The polarization of the  Nijenhuis torsion of level $m$ is the vector-valued 2--form 
\begin{equation} \label{eq:pol}
\begin{split}
\mathcal{P}^{(m)} (\bs{A}_1,\bs{A}_2,\ldots,\bs{A}_{2m})(X,Y)&:= \Delta_{2m}^{(m)} (\bs{A}_1,\bs{A}_2,\ldots,\bs{A}_{2m})(X,Y) \\
&=
%\frac{1}{m!} 
\sum_{\mathcal{I} \subseteq \{1,2,\ldots,2m\}} (-1)^{2m-|\mathcal{I}| }\ \mathcal{\tau}^{(m)}_{{\sum_{i\in \mathcal{I}} \bs{A}_{i}}}(X,Y) 
\end{split}
\end{equation}
where $m\in \mathbb{N}\setminus \{0\}$ and $\bs{A}_1,\ldots,\bs{A}_{2m} \in \mathcal{T}^{1}_{1}(M)$.
\end{definition}
\begin{remark}
We omit the factor $\frac{1}{(2m)!}$ usually adopted in polarization formulae (see e.g. \cite{Land}) in order to recover exactly the Fr\"olicher-Nijenhuis bracket \eqref{eq:binNtors} for $m=1$:
\begin{equation} \label{eq:FN=P1}
\mathcal{P}^{(1)}(\bs{A}_1,\bs{A}_2)(X,Y)= \llbracket\bs{A}_1,\bs{A}_2\rrbracket(X,Y) \ .
\end{equation}

\end{remark}

\noi By way of an example, let us consider the 
case $m=2$. The \textit{polarization of the Haantjes torsion} reads explicitly:
\begin{eqnarray}
&& \nn\mathcal{P}^{(2)} (\bs{A}_1,\bs{A}_2,\bs{A}_3, \bs{A}_4)(X,Y) =
\mathcal{H} (\bs{A}_1+\bs{A}_2+\bs{A}_3+\bs{A}_4)(X,Y) \\ \nn
&-&
\bigg(\mathcal{H} (\bs{A}_1+\bs{A}_2+\bs{A}_3) + \mathcal{H} (\bs{A}_1+\bs{A}_2+\bs{A}_4)+
\mathcal{H} (\bs{A}_1+\bs{A}_3+\bs{A}_4)  +\mathcal{H} (\bs{A}_2+\bs{A}_3+\bs{A}_4)\bigg)(X,Y)\\ \nn
&+& \bigg(\mathcal{H} (\bs{A}_1+\bs{A}_2)+  \mathcal{H} (\bs{A}_1+\bs{A}_3)+ \mathcal{H} (\bs{A}_1+\bs{A}_4)+
 \mathcal{H} (\bs{A}_2+\bs{A}_3) + \mathcal{H} (\bs{A}_2+\bs{A}_4)+ \mathcal{H} (\bs{A}_3+\bs{A}_4)\bigg)(X,Y)\\ \nn
 &-& \Big(\mathcal{H} (\bs{A}_1) + \mathcal{H} (\bs{A}_2) + \mathcal{H} (\bs{A}_3) +\mathcal{H} (\bs{A}_4)\Big)(X,Y) \ ,
\end{eqnarray}
where the Haantjes torsion $\mathcal{H} (\bs{A})(X,Y):= \mathcal{H}_{\bs{A}}(X,Y)$ is defined by Eq. \eqref{eq:Haan}.

\vspace{3mm}

Following the approach of \cite{Proc2007}, we deduce an equivalent, useful definition of polarization: 
\begin{equation}
\mathcal{P}^{(m)}(\bs{A}_1,\bs{A}_2,\ldots,\bs{A}_{2m})(X,Y)=\frac{\partial^{2m}}{ \partial \lambda_1\dots \partial\lambda_{2m}} \mathcal{\tau}^{(m)}_{{\sum_{i=1}^{2m}\lambda_i \bs{A}_{i}}}(X,Y) \ ,
\end{equation}
 the $\lambda_i$ being distinct arbitrary real parameters.
\begin{lemma}
The following relation holds
\begin{equation}
\Delta_1^{(m)}(\bs{A})(X,Y)=\tau_{\bs{A}}^{(m)}(X,Y)=\frac{1}{(2m)!}\ \mathcal{P}^{(m)}(\underbrace{\bs{A}, \ldots,\bs{A}}_{2m-\text{times}})
\end{equation}
\begin{proof}
\begin{equation}
\begin{split}
&\mathcal{P}^{(m)} (\underbrace{\bs{A}, \ldots,\bs{A}}_{2m-\text{times}})(X,Y)=\frac{\partial^{2m}}{ \partial \lambda_1\dots \partial\lambda_{2m}} \mathcal{\tau}^{(m)}_{{\sum_{i=1}^{2m}\lambda_i \bs{A}}}(X,Y) =
\frac{\partial^{2m}}{ \partial \lambda_1\dots \partial\lambda_{2m}} (\sum_{i=1}^{2m}\lambda_i)^{2m} \mathcal{\tau}^{(m)}_{{\bs{A}}}(X,Y) 
\\
& = (2m)! \,  \mathcal{\tau}^{(m)}_{{\bs{A}}}(X,Y)
\end{split}
\end{equation}
thanks to the identity
$$
\frac{\partial^{n}}{ \partial \lambda_1\dots \partial\lambda_{n}} (\sum_{i=1}^{n}\lambda_i)^{n} = n!
$$
that is proved by induction over $n$.
\end{proof}
\end{lemma}
Let us prove a result which will be used in the sequel of our analysis.
\begin{theorem} \label{th:3.7}
The polarization of the Nijenhuis torsion of level $m$ satisfies the following recurrence relation:
\begin{equation} \label{eq:polRec}
\begin{split}
\mathcal{P}^{(m)} (\bs{A}_1,\bs{A}_2,\ldots,\bs{A}_{2m})(X,Y)=
\sum_{ \substack{
                  j,k,i_1,\ldots, i_{2m-2} \neq \\
                     i_1<\ldots <i_{2m-2} } }^{2m}
          (\bs{A}_j\bs{A}_k+\bs{A}_k\bs{A}_j) \mathcal{P}^{(m-1)} (\bs{A}_{i_1},\ldots,\bs{A}_{i_{2m-2}})(X,Y)\\
 +
 \sum_{ \substack{
                  j,k,i_1,\ldots, i_{2m-2} \neq \\
                     i_1<\ldots <i_{2m-2} } }^{2m} \bigg ( \mathcal{P}^{(m-1)} (\bs{A}_{i_1},\ldots,\bs{A}_{i_{2m-2}})(\bs{A}_j X,\bs{A}_k Y)+
 \mathcal{P}^{(m-1)} (\bs{A}_{i_1},\ldots,\bs{A}_{i_{2m-2}}) (\bs{A}_kX,\bs{A}_j Y) \bigg )\\ 
  -
  \sum_{ \substack{
                  j,k,i_1,\ldots, i_{2m-2} \neq \\
                     i_1<\ldots <i_{2m-2} } }^{2m}   \bs{A}_j \bigg( \mathcal{P}^{(m-1)} (\bs{A}_{i_1},\ldots,\bs{A}_{i_{2m-2}}) (\bs{A}_k X,Y)+
 \mathcal{P}^{(m-1)} (\bs{A}_{i_1},\ldots,\bs{A}_{i_{2m-2}}) (X,\bs{A}_kY) \bigg ) 
\end{split}
\end{equation}
where $m\in \mathbb{N}\setminus \{0,1 \}$, and  is $\mathbb{R}$-multilinear, that is, it is  $\mathbb{R}$-linear in each of the $\bs{A}_i\in \mathcal{T}^{1}_{1}(M)$ for  $m\in \mathbb{N}\setminus \{0 \}$. Here the symbol $j,k,i_1,\ldots, i_{2m-2} \neq$ means that no index is repeated in the sums. 
\end{theorem}
\begin{proof}
Both statements can be proved by induction over $m$.
\par
  For $m=1$, we obtain  the relation $\mathcal{P}^{(1)}(\bs{A}_1,\bs{A}_2)=  \llbracket\bs{A}_1,\bs{A}_2\rrbracket $ which is  $\mathbb{R}$-bilinear (see Definition \ref{eq:binNtors}).
\par
For $m=2$, after a lengthy calculation,  we get 
\begin{eqnarray} \label{eq:Pol2}
\nn&&  \mathcal{P}^{(2)} (\bs{A}_1,\bs{A}_2,\bs{A}_3, \bs{A}_4)(X,Y) =
 \\ \nn
 &=&\sum_{ \substack{
                  j,k,i_1, i_{2} \neq \\
                     i_1<i_{2} } } ^4(\bs{A}_j\bs{A}_k+\bs{A}_k\bs{A}_j) \llbracket \boldsymbol{A}_{i_1},\bs{A}_{i_2}\rrbracket (X,Y) 
\\ \nn                     
 &+&
 \sum_{ \substack{
                  j,k,i_1, i_{2} \neq \\
                     i_1<i_{2} } }^4\bigg ( \llbracket \boldsymbol{A}_{i_1},\bs{A}_{i_2}\rrbracket (\bs{A}_j X,\bs{A}_kY)+
  \llbracket \boldsymbol{A}_{i_1},\bs{A}_{i_2}\rrbracket (\bs{A}_k X,\bs{A}_j Y) \bigg )\\ \nn
&  -&
   \sum_{ \substack{
                  j,k,i_1, i_{2} \neq \\
                     i_1<i_{2} } }^4  \bs{A}_j \bigg( \llbracket \boldsymbol{A}_{i_1},\bs{A}_{i_2}\rrbracket (\bs{A}_{k} X,Y)+
  \llbracket \boldsymbol{A}_{i_1},\bs{A}_{i_2}\rrbracket (X,\bs{A}_k Y) \bigg ) \\
\end{eqnarray}
which is a special case of Eq. \eqref{eq:polRec} in view of Eq. \eqref{eq:FN=P1}. Indeed, Eq. \eqref{eq:Pol2} implies that 
$ \mathcal{P}^{(2)} (\bs{A}_1,\bs{A}_2,\bs{A}_3, \bs{A}_4)$ is  $\mathbb{R}$-linear in $\bs{A}_1,\bs{A}_2,\bs{A}_3,\bs{A}_4$ thanks to the bilinearity of the standard Fr\"olicher-Nijenhuis bracket \eqref{eq:binNtors}.
%For the $\mathbb{R}$-linearity, observe that for $m=1$, $\mathcal{P}^{(1)}(\bs{A}_1,\bs{A}_2)=  \llbracket\bs{A}_1,\bs{A}_2\rrbracket $ which is  $\mathbb{R}$-linear. 
\par
Let us assume as induction hypotheses that the recursion relation  \eqref{eq:polRec}  holds for  the polarization of the Nijenhuis torsion of level $m-1$ and  that this polarization  is linear in $\bs{A}_1,\ldots,\bs{A}_{2m-2}$. Let us   show that  Eq. \eqref{eq:polRec} holds true for the polarized Nijenhuis torsion of level $m$. 
\begin{equation} \label{eq:pol}
\begin{split}
&\mathcal{P}^{(m)} (\bs{A}_1,\bs{A}_2,\ldots,\bs{A}_{2m})(X,Y)=\frac{\partial^{2m}}{ \partial \lambda_1\dots \partial\lambda_{2m}}
 \mathcal{\tau}^{(m)}_{{\sum_{i=1}^{2m}\lambda_i \bs{A}_{i}}}(X,Y) \\
&=\frac{\partial^{2m}}{ \partial \lambda_1\dots \partial\lambda_{2m}}
\bigg( \big(\sum_{i=1}^{2m} \lambda_i \bs{A}_i)^2  \mathcal{\tau}^{(m-1)}_{{\sum_{i=1}^{2m}\lambda_i \bs{A}_{i}}}(X,Y)  +
\mathcal{\tau}^{(m-1)}_{{\sum_{i=1}^{2m}\lambda_i \bs{A}_{i}}}
\Big(\sum_{i=1}^{2m} \lambda_i\bs{A}_i X, \sum_{i=1}^{2m} \lambda_i\bs{A}_i Y \Big)
\\
 &-\Big(\sum_{i=1}^{2m} \lambda_i\bs{A}_i\Big )  \Big(\mathcal{\tau}^{(m-1)}_{{\sum_{i=1}^{2m}\lambda_i \bs{A}_{i}}}
\Big(\sum_{i=1}^{2m} \lambda_i\bs{A}_i X, Y \Big) +
\mathcal{\tau}^{2m}_{{\sum_{i=1}^{2m}\lambda_i \bs{A}_{i}}}
\Big( X, \sum_{i=1}^{2m} \lambda_i\bs{A}_i Y \Big)\bigg)
\end{split}
\end{equation}
If we consider the first term and the induction hypothesis about  the linearity of $\mathcal{P}^{(m-1)}$ we get
\begin{equation}
\begin{split}
&\frac{\partial^{2m}}{ \partial \lambda_1\dots \partial\lambda_{2m}}
\bigg( \big(\sum_{i=1}^{2m} \lambda_i \bs{A}_i)^2  \mathcal{\tau}^{(m-1)}_{{\sum_{i=1}^{2m}\lambda_i \bs{A}_{i}}}(X,Y)\bigg)\\
&=\frac{1}{(2m-2)!} \frac{\partial^{2m}}{ \partial \lambda_1\dots \partial\lambda_{2m}}
 \bigg(\big(\sum_{i=1}^{2m} \lambda_i \bs{A}_i)^2  \mathcal{P}^{(m-1)}(\underbrace{\sum_{i=1}^{2m}\lambda_i \bs{A}_{i},\ldots,\sum_{i=1}^{2m}\lambda_i \bs{A}_{i}}_{(2m-2)-\text{times}})(X,Y)\bigg)\\ 
&=\frac{1}{(2m-2)!} \frac{\partial^{2m}}{ \partial \lambda_1\dots \partial\lambda_{2m}}
 \bigg(\big(\sum_{i=1}^{2m} \lambda_i \bs{A}_i)^2 \sum_{i_1,\ldots,i_{2m-2}}^{2m} \lambda_{i_1}\ldots\lambda_{i_{2m-2}}\mathcal{P}^{(m-1)}(\bs{A}_{i_{1}}, \ldots,\bs{A}_{i_{2m-2}})(X,Y)\bigg) \\
 \end{split}
\end{equation}
\begin{equation*}
\begin{split}
 &=\frac{1}{(2m-2)!} \frac{\partial^{2m}}{ \partial \lambda_1\dots \partial\lambda_{2m}}
 \bigg(\sum_{i_1,\ldots,i_{2m}}^{2m} \lambda_{i_1}\ldots\lambda_{i_{2m}} \bs{A}_{i_{2m-1}} \bs{A}_{i_{2m}} \mathcal{P}^{(m-1)}(\bs{A}_{i_{1}}, \ldots,\bs{A}_{i_{2m-2}})(X,Y)\bigg)
 \\ 
 &=
% \frac{1}{\cancel{(2m-2)!}}
  \frac{\partial^{2m}}{ \partial \lambda_1\dots \partial\lambda_{2m}}
 \bigg(\lambda_{1} \ldots \lambda_{{2m}} \sum_{ \substack{
                  j,k,i_1,\ldots, i_{2m-2} \neq \\
                     i_1<\ldots <i_{2m-2} } }^{2m} (\bs{A}_{j}\bs{A}_{k}+\bs{A}_{k}  \bs{A}_{j})
% \cancel{ (2m-2)!}
\mathcal{P}^{(m-1)}(\bs{A}_{i_{1}}, \ldots,\bs{A}_{i_{2m-2}})(X,Y)\bigg) \ ,
\end{split}
\end{equation*}
therefore we obtain the first summation in the right term of Eq. \eqref{eq:polRec}.
\par
Analogously, considering the second term in Eq. \eqref{eq:pol} we get 
\begin{equation}
\begin{split}
&\frac{\partial^{2m}}{ \partial \lambda_1\dots \partial\lambda_{2m}}
\bigg(\mathcal{\tau}^{(m-1)}_{{\sum_{k=1}^{2m}\lambda_k \bs{A}_{k}}}
\Big(\sum_{i=1}^{2m} \lambda_i\bs{A}_i X, \sum_{j=1}^{2m} \lambda_j\bs{A}_j Y \Big)
\\
&=\frac{\partial^{2m}}{ \partial \lambda_1\dots \partial\lambda_{2m}}
\bigg( \sum_{i,j=1}^{2m} \lambda_i \lambda_j \, \mathcal{\tau}^{(m-1)}_{{\sum_{i=1}^{2m}\lambda_i \bs{A}_{i}}}
\Big(\bs{A}_i X, \bs{A}_j Y \Big)\bigg)
\\
&=\frac{1}{(2m-2)!} \frac{\partial^{2m}}{ \partial \lambda_1\dots \partial\lambda_{2m}}
 \bigg(\sum_{i,j=1}^{2m} \lambda_i \lambda_j \, \mathcal{P}^{(m-1)}(\underbrace{\sum_{i=1}^{2m}\lambda_k \bs{A}_{k},\ldots,\sum_{i=1}^{2m}\lambda_k \bs{A}_{k}}_{(2m-2)-\text{times}})(\bs{A}_i X,\bs{A}_jY)\bigg)
 \end{split}
\end{equation}
\begin{equation*}
\begin{split}
  &=\frac{1}{(2m-2)!} \frac{\partial^{2m}}{ \partial \lambda_1\dots \partial\lambda_{2m}}
 \bigg(\sum_{i_1,\ldots,i_{2m}}^{2m} \lambda_{i_1}\ldots\lambda_{i_{2m}} \mathcal{P}^{(m-1)}(\bs{A}_{i_{1}}, \ldots,\bs{A}_{i_{2m-2}})( \bs{A}_{i_{2m-1}} X,\bs{A}_{i_{2m}}Y)\bigg)
 \\ 
 &=
  \frac{\partial^{2m}}{ \partial \lambda_1\dots \partial\lambda_{2m}}
 \bigg(\lambda_{1} \ldots \lambda_{{2m}} \sum_{ \substack{
                  j,k,i_1,\ldots, i_{2m-2} \neq \\
                     i_1<\ldots <i_{2m-2} } }^{2m} 
          \Big(   \mathcal{P}^{(m-1)}(\bs{A}_{i_{1}}, \ldots,\bs{A}_{i_{2m-2}}) (\bs{A}_j X,\bs{A}_k Y) \\
&+ \mathcal{P}^{(m-1)}(\bs{A}_{i_{1}}, \ldots,\bs{A}_{i_{2m-2}}) (\bs{A}_k X,\bs{A}_j Y) \Big)
                                         \bigg) \ .
\end{split}
\end{equation*}
Thus, we have obtained the second summation in the right term of Eq. \eqref{eq:polRec}.
\par
In the same manner, from the third term of Eq. \eqref{eq:pol} we get 
\begin{equation}
\begin{split}
&\frac{\partial^{2m}}{ \partial \lambda_1\dots \partial\lambda_{2m}}
\bigg( \Big(\sum_{j=1}^{2m} \lambda_j\bs{A}_j\Big )  \Big(\mathcal{\tau}^{(m-1)}_{{\sum_{i=1}^{2m}\lambda_i \bs{A}_{i}}}
\Big(\sum_{k=1}^{2m} \lambda_k\bs{A}_k X, Y \Big) +
\mathcal{\tau}^{2m}_{{\sum_{i=1}^{2m}\lambda_i \bs{A}_{i}}}
\Big( X, \sum_{k=1}^{2m} \lambda_k\bs{A}_k Y \Big)\bigg)
\\
&=\frac{\partial^{2m}}{ \partial \lambda_1\dots \partial\lambda_{2m}}
\bigg( \Big(\sum_{j=1}^{2m} \lambda_j\bs{A}_j\Big )  
\Big( \sum_{k=1}^{2m} \lambda_k \Big( \mathcal{\tau}^{(m-1)}_{\sum_{i=1}^{2m} \lambda_i\bs{A}_{i}}
\Big(\bs{A}_k X, Y \Big) + \mathcal{\tau}^{(m-1)}_{\sum_{i=1}^{2m} \lambda_i\bs{A}_{i}}
\Big( X, \bs{A}_k Y \Big)\Big)\bigg)
\end{split}
\end{equation}
\begin{equation*}
\begin{split}
&=\frac{1}{(2m-2)!} \frac{\partial^{2m}}{ \partial \lambda_1\dots \partial\lambda_{2m}}
 \bigg(\sum_{j,k=1}^{2m} \lambda_j \lambda_k \bs{A}_j \Big(  \mathcal{P}^{(m-1)}(\underbrace{\sum_{i=1}^{2m}\lambda_i \bs{A}_{i},\ldots,\sum_{i=1}^{2m}\lambda_i \bs{A}_{i}}_{(2m-2)-\text{times}})(\bs{A}_k X,Y)\\ 
 &+\mathcal{P}^{(m-1)}(\underbrace{\sum_{i=1}^{2m}\lambda_i \bs{A}_{i},\ldots,\sum_{i=1}^{2m}\lambda_i \bs{A}_{i}}_{(2m-2)-\text{times}})(X,\bs{A}_k Y) \Big )\bigg)
 \end{split}
\end{equation*}
 \begin{equation*}
\begin{split}
  &=\frac{1}{(2m-2)!} \frac{\partial^{2m}}{ \partial \lambda_1\dots \partial\lambda_{2m}}
 \bigg(\sum_{j,k,i_1,\ldots,i_{2m-2}}^{2m} \lambda_j \lambda_k \lambda_{i_1}\ldots\lambda_{i_{2m-2}} \bs{A}_j\Big(\mathcal{P}^{(m-1)}(\bs{A}_{i_{1}}, \ldots,\bs{A}_{i_{2m-2}})( \bs{A}_k X,Y)
 \\ 
 &+ 
  \mathcal{P}^{(m-1)}(\bs{A}_{i_{1}}, \ldots,\bs{A}_{i_{2m-2}})( X,\bs{A}_k Y)\Big)\bigg)
 \end{split}
\end{equation*}
 \begin{equation*}
\begin{split}
 &=
  \frac{\partial^{2m}}{ \partial \lambda_1\dots \partial\lambda_{2m}}
 \bigg(\lambda_{1} \ldots \lambda_{{2m}} \sum_{ \substack{
                  j,k,i_1,\ldots, i_{2m-2} \neq \\
                     i_1<\ldots <i_{2m-2} } }^{2m} 
         \bs{A}_j \Big(  \, \mathcal{P}^{(m-1)}(\bs{A}_{i_{1}}, \ldots,\bs{A}_{i_{2m-2}}) (\bs{A}_k X, Y) \\
&+ \mathcal{P}^{(m-1)}(\bs{A}_{i_{1}}, \ldots,\bs{A}_{i_{2m-2}}) (X,\bs{A}_k Y) \Big)
                                         \bigg) \ .
\end{split}
\end{equation*}
In this way, we have deduced the third summation in the right term of Eq. \eqref{eq:polRec}. Concerning $\mathbb{R}$-multilinearity, it is easily deduced by assuming that the property holds for each polarization of level $m-1$ in Eq. \eqref{eq:polRec} (all of them involving $2m-2$ arguments), and taking also into account that in each summand the explicit dependence on the remaining arguments $\bs{A}_j$ and $\bs{A}_k$ ($j\neq k$) is linear. 
\end{proof}
%By construction, the polarized Haantjes torsion is symmetric in each of the operators $\bs{A}_i$. Another essential property is $\mathbb{R}$-linearity.
\begin{proposition} \label{pr:linear}
%The polarized Nijenhuis torsion of level $m$ \eqref{eq:pol} for $m\geq 1$ is $\mathbb{R}$-linear in each of the $\bs{A}_i\in \mathcal{T}^{1}_{1}(M)$. Moreover,
 If the operators $\bs{A}_i\in \mathcal{T}^{1}_{1}(M)$ pairwise commute, then for $m\geq 2$ the polarization of the Nijenhuis torsion is $C^{\infty}(M)$-multilinear.
\end{proposition}
\begin{proof}
Once again, this  statement can be proved by induction over $m$.
\par

 In view of Eq. \eqref{eq:FNf} we obtain
\begin{eqnarray}
&& \nn\mathcal{P}^{(2)} (f \bs{A}_1,\bs{A}_2,\bs{A}_3, \bs{A}_4)(X,Y) =
f \mathcal{P}^{(2)} (f \bs{A}_1,\bs{A}_2,\bs{A}_3, \bs{A}_4)(X,Y) +
\\
 \nn
&-&
\bigg( \bs{A}_2 [\bs{A}_3, \bs{T}(\bs{A}_4,\bs{A}_1)-\bs{T}^T(\bs{A}_4,\bs{A}_1)  ]+\bs{A}_3    [\bs{A}_2, \bs{T}(\bs{A}_4,\bs{A}_1)-\bs{T}^T(\bs{A}_4,\bs{A}_1)  ]\bigg)  (\rd f,X,Y)\\
&-& \nn
\bigg( \bs{A}_2 [\bs{A}_4, \bs{T}(\bs{A}_3,\bs{A}_1)-\bs{T}^T(\bs{A}_3,\bs{A}_1)  ]+\bs{A}_4    [\bs{A}_2, \bs{T}(\bs{A}_3,\bs{A}_1)-\bs{T}^T(\bs{A}_3,\bs{A}_1)  ]\bigg)  (\rd f, X,Y)\\ 
&-& \nn
\bigg( \bs{A}_3 [\bs{A}_4, \bs{T}(\bs{A}_2,\bs{A}_1)-\bs{T}^T(\bs{A}_2,\bs{A}_1)  ]+\bs{A}_4    [\bs{A}_3, \bs{T}(\bs{A}_2,\bs{A}_1)-\bs{T}^T(\bs{A}_2,\bs{A}_1)  ] \bigg) (\rd f,X,Y)
\end{eqnarray}
\begin{eqnarray*}
&+&\nn
\bigg([\bs{A}_3, \bs{T}(\bs{A}_4,\bs{A}_1)-\bs{T}^T(\bs{A}_4,\bs{A}_1)  ]+   [\bs{A}_4, \bs{T}(\bs{A}_3,\bs{A}_1)-\bs{T}^T(\bs{A}_3,\bs{A}_1)  ] \bigg) (\rd f, \bs{A}_2 X,Y)\\ 
&+&\nn
\bigg([\bs{A}_4, \bs{T}(\bs{A}_2,\bs{A}_1)-\bs{T}^T(\bs{A}_2,\bs{A}_1)  ]+   [\bs{A}_2, \bs{T}(\bs{A}_4,\bs{A}_1)-\bs{T}^T(\bs{A}_4,\bs{A}_1)  ] \bigg) (\rd f, \bs{A}_3 X,Y)\\ 
&+&\nn
\bigg([\bs{A}_2, \bs{T}(\bs{A}_3,\bs{A}_1)-\bs{T}^T(\bs{A}_3,\bs{A}_1)  ]+   [\bs{A}_3, \bs{T}(\bs{A}_2,\bs{A}_1)-\bs{T}^T(\bs{A}_2,\bs{A}_1)  ] \bigg) (\rd f, \bs{A}_4 X,Y)\\ 
\end{eqnarray*}
where $ \bs{T}$ is defined by Eq. \eqref{eq:T}.
As the operators $\bs{A}_i$ by assumption commute pairwise, each of the commutators of the form
$$
[\bs{A}_i, \bs{T}(\bs{A}_j,\bs{A}_k)-\bs{T}^T(\bs{A}_j,\bs{A}_k)  ] \qquad\qquad i,j,k \neq
$$
vanishes. Therefore $\mathcal{P}^{(2)} ( \bs{A}_1,\bs{A}_2,\bs{A}_3, \bs{A}_4)$ is $C^{\infty}(M)$-linear. Again
by induction over $m$ and taking into account Eq. \eqref{eq:polRec}, we deduce the $C^{\infty}(M)$-linearity of the polarization of the Nijenhuis  torsions for any level $m\geq 2$.
\end{proof}
\begin{proposition}
For $k>2m$, we have
\beq
\Delta_k^{(m)} (\bs{A}_1,\bs{A}_2,\ldots,\bs{A}_{k})(X,Y)=0 \ .
\eeq
\end{proposition}
\begin{proof}
It is a simple consequence of  Proposition \ref{pr:Deltanull}, Theorem \ref{th:3.7} and the recursion relations \eqref{eq:DeltaRec}.
\end{proof}
Observe that the defect of index $2$ for any $m$ has the simple form 
\begin{equation}\label{eq:Delta2}
\Delta_2^{(m)}(\bs{A},\bs{B})(X,Y)=\tau_{\bs{A}+\bs{B}}^{(m)}(X,Y)-\tau_{\bs{A}}^{(m)}(X,Y)-\tau_{\bs{B}}^{(m)}(X,Y)\ .
\end{equation}
Thus, we can state the following result.
\begin{lemma}
The relation 
\begin{equation} \label{eq:Delta2P}
\Delta_2^{(m)}(\bs{A},\bs{B})(X,Y)=
 \frac{1}{(2m)!} \bigg(\sum_{k=1}^{2m-1}  \binom{2m}{k} \mathcal{P}^{(m)}(\hspace{-2mm}\underbrace{\bs{A},\ldots,\bs{A}}_{(2m-k)-\text{times}}\hspace{-2mm},\, \underbrace{\bs{B},\ldots,\bs{B}}_{k-\text{times}}) (X,Y) \bigg)
\end{equation}
holds.
\end{lemma}
\begin{proof}
Using eq. \eqref{eq:Delta2} we have
\begin{eqnarray}
&& \hspace{5mm} \big(\tau_{\bs{A}+\bs{B}}^{(m)}-\tau_{\bs{A}}^{(m)}-\tau_{\bs{B}}^{(m)}\Big)(X,Y) = \\ \nn
&&\frac{1}{(2m)!}\Big(\mathcal{P}^{(m)} (\underbrace{\bs{A+B},\ldots,\bs{A+B}}_{2m-\text{times}})-\mathcal{P}^{(m)} (\underbrace{\bs{A},\ldots,\bs{A}}_{2m-\text{times}})-\mathcal{P}^{(m)} (\underbrace{\bs{A},\ldots,\bs{A}}_{2m-\text{times}})\Big)(X,Y)
\end{eqnarray}
In view of the additivity property of $\mathcal{P}^{(m)}$, by induction on $m$ it can be proved that 
$$
\mathcal{P}^{(m)} (\underbrace{\bs{A+B},\ldots,\bs{A+B}}_{2m-\text{times}})(X,Y)=\sum_{k=0}^{2m} \binom{2m}{k}
 \mathcal{P}^{(m)} (\!\!\! \! \underbrace{\bs{A},\ldots,\bs{A}}_{(2m-k)-\text{times}}\hspace{-3mm},\,
 \underbrace{\bs{B},\ldots,\bs{B}}_{k-\text{times}}) (X,Y) \ ,
$$
which completes the proof.
\end{proof}
\subsection{Generalized Haantjes modules and vector spaces}
In the following analysis, we shall further illustrate the relevance of the polarizations of the generalized Nijenhuis torsion \eqref{eq:pol}. Indeed they play a crucial role in the study of the   $C^{\infty}(M)$-modules of generalized Nijenhuis operators. In \cite{TT2021JGP}, we have introduced the notion of \textit{Haantjes modules}, namely  modules of Haantjes operators. This notion has been further generalized in the recent work \cite{RTT2022arx}, where the generalized Haantjes algebras of level $m$ have been introduced. Here we are interested in the case of generalized modules of operators, according to the following
\begin{definition}\label{def:HM}
A generalized Haantjes module of level $m$ is a pair    $(M, \mathscr{H}^{(m)}_{\mathcal{M}})$ which satisfies the following  conditions:
\begin{itemize}
\item
$M$ is a differentiable manifold of dimension $n$;
\item
$\mathscr{H}^{(m)}_{\mathcal{M}}$ is a set of   operators $\boldsymbol{K}:\mathfrak{X}(M)\rightarrow \mathfrak{X}(M)$   such that
\begin{equation}\label{eq:Hmod}
\tau^{(m)}_{f\boldsymbol{K_1} +
                             g\boldsymbol{K}_2}(X,Y)= \boldsymbol{0}
 \ , \qquad\forall X, Y \in \mathfrak{X}(M) \ , \quad \forall f,g \in C^\infty(M)\  ,\quad \forall \boldsymbol{K}_1,\boldsymbol{K}_2 \in  \mathscr{H}^{(m)}_{\mathcal{M}}.
\end{equation}
\end{itemize}
\end{definition}
Thus, a generalized Haantjes module is a free module of generalized Nijenhuis operators of level $m$ over the ring of smooth functions on $M$. In particular, if property \eqref{eq:Hmod} is satisfied  when $f,g$ are  real constants, we shall use the denomination of \textit{generalized Haantjes vector space of level $m$}.

\vspace{3mm}

Our main goal is to determine  the tensorial compatibility conditions ensuring the existence of the generalized Haantjes module of level $m$ generated by \textit{two} operators $\bs{A}$, $\bs{B}:\mathfrak{X}(M)\to \mathfrak{X}(M)$. We construct these conditions in the most general case of not necessarily semisimple operators, under the unique hypothesis  that they commute.

%Then, we shall restrict to the important case of semisimple, commuting operators, which arises for instance in Hamiltonian classical mechanics, in the discussion of integrable systems  \cite{TT2021AMPA}, \cite{RTT2022CNS}. 

%Precisely, for semisimple commuting operators we shall prove that the condition $\mathcal{H}_{\bs{A},\bs{B}}(X,Y)=\bs{0}$ is necessary and sufficient for the existence of the Haantjes module generated by $\bs{A}$ and $\bs{B}$.
%\subsubsection{The general case}
We shall start our analysis from
the following identity, valid for all  $X,Y$$\in$ $\mathfrak{X}(M)$:
\begin{equation} \label{eq:HaanA+B}
\tau^{(m)}_{\lambda\bs{A}+\mu\bs{B}}(X,Y)=\lambda^{2m}\tau_{\boldsymbol{A}}^{(m)}(X,Y)+\mu^{2m}\tau_{\boldsymbol{B}}^{(m)}(X,Y)+
\Delta_2^{(m)}(\lambda\bs{A},\mu\bs{B})(X,Y) \ .
\end{equation}
For $m\geq 1$, this identity holds  $\forall\lambda,\mu$ $\in$ $\mathbb{R}$ as a consequence of Eqs. \eqref{eq:Delta2} and \eqref{tauind}.  For $m\geq 2$, it holds more generally $\forall \lambda,\mu$ $\in C^{\infty}(M)$ due to Eqs. \eqref{eq:Delta2} and \eqref{tauind2}.
\begin{theorem} \label{th:main}
Let $M$ be a differentiable manifold and let $\boldsymbol{A},\boldsymbol{B}:\mathfrak{X}(M)\rightarrow \mathfrak{X}(M)$ be two generalized Nijenhuis operators of level $m$. Then for $m\geq 1$, $\bs{A}$ and $\bs{B}$ generate a  Haantjes vector space of  level $m$  if and only if the $(2m-1)$ differential  conditions 
\beq \label{eq:DC}
  \mathcal{P}^{(m)}(\hspace{-3mm}\underbrace{\bs{A},\ldots\bs{A}}_{(2m-k)-\text{times}}\hspace{-2mm},
 \underbrace{\bs{B},\ldots,\bs{B}}_{k-\text{times}}) (X,Y) = \bs{0} \ , \qquad k=1,\ldots,2m-1
\eeq
are satisfied. In addition, for $m\geq 2$,  if $\bs{A}$ and $\bs{B}$ commute, they generate a $C^{\infty}(M)$ Haantjes module  of level $m$   if and only if the same  conditions \eqref{eq:DC} are fulfilled.
\end{theorem}
\begin{proof}
From Eqs.  \eqref{eq:HaanA+B} it follows that $\bs{A}$ and $\bs{B}$ generate  a generalized Haantjes vector space (resp.  a generalized Haantjes module) of level $m$ if and only if the defect $\Delta_2^{(m)}(\lambda\bs{A},\mu\bs{B}) $ vanishes for all $\lambda$ and $\mu \in \mathbb{R}$ (resp. $\in C^\infty(M)$). From Eq.  \eqref{eq:Delta2P} and    by the multilinearity of the polarization of the Nijenhuis  torsion of level $m$, such a defect takes the form
\begin{equation} \label{eq:Delta2}
\Delta_2^{(m)}(\lambda\bs{A},\mu\bs{B})(X,Y)=
 \frac{1}{(2m)!} \bigg( \sum_{k=1}^{2m-1}  \binom{2m}{k} \lambda^{(2m-k)}\mu^k\mathcal{P}^{(m)}(\hspace{-3mm}\underbrace{\bs{A},\ldots\bs{A}}_{(2m-k)-\text{times}}\hspace{-2mm}, \underbrace{\bs{B},\ldots,\bs{B}}_{k -\text{times}}) (X,Y) \bigg) \ ,
\end{equation}
for $m\geq 1$ and $\lambda , \mu \in \mathbb{R}$, and for $m\geq2$ and $\lambda, \mu \in  C^\infty(M)$ when $\bs{A}$ and $\bs{B}$ commute. From Eq. \eqref{eq:Delta2},  conditions \eqref{eq:DC} follow.
%We can use the same reasoning valid for the proof of Proposition \ref{pr:linear} to conclude that for $m\geq 1$ the $\mathbb{R}$-linearity property holds for any choice of $\bs{A}$ and $\bs{B}$, whereas the $C^\infty(M)$-linearity holds for $m\geq2$ if $\bs{A}$ and $\bs{B}$  commute. 
\end{proof}

\subsection{Higher Haantjes  brackets}
Let $M$ be a differentiable manifold and $\boldsymbol{A},\boldsymbol{B}:\mathfrak{X}(M)\rightarrow \mathfrak{X}(M)$ two operators belonging to $\mathcal{T}_1^1(M)$.

%\subsection{Higher brackets}
%Hereafter, we shall present the main algebraic construction of this work, namely the recursive definition of an infinite ``tower'' of new brackets of couples of  operators.
\begin{definition}[\cite{TT2022CMP}] \label{GHBT}
%Let $M$ be a differentiable manifold of dimension $n$ and let $\boldsymbol{A},\boldsymbol{B}:\mathfrak{X}(M)\rightarrow \mathfrak{X}(M)$ be two $(1,1)$ tensors.
 The
 \textit{Haantjes bracket of level $m\in\mathbb{N}\backslash\{0\}$} of $\boldsymbol{A}$ and $\boldsymbol{B}$  is the vector--valued $2$--form defined, for any  $X,Y \in \mathfrak{X}(M)$, by the relations
\bea \label{GHTn}
\nn & & \!\!\!\!\!\!\!\!\mathcal{H}^{(1)}_{\boldsymbol{A,B}} (X,Y) := \llbracket\bs{A}, \bs{B}\rrbracket(X,Y) \\
\nn \text{and}\\
\nn & &\!\!\!\!\!\!\!\!  \mathcal{H}^{(m)}_{\boldsymbol{A,B}}(X,Y):=\Big(\bs{AB}+ \bs{BA}\Big)\mathcal{H}_{\boldsymbol{A,B}}^{(m-1)}(X,Y) +\mathcal{H}_{\boldsymbol{A,B}}^{(m-1)}(\bs{A}X,\bs{B}Y)+\mathcal{H}_{\boldsymbol{A,B}}^{(m-1)}(\bs{B}X,\bs{A}Y)
 \\
\nn &&\!\!\!\!\!\!\!\! -\bs{A}\Big(\mathcal{H}_{\boldsymbol{A,B}}^{(m-1)}(X,\bs{B}Y) +\mathcal{H}_{\boldsymbol{A,B}}^{(m-1)}(\bs{B}X,Y)\big)-\bs{B}\big(\mathcal{H}_{\boldsymbol{A,B}}^{(m-1)}(X,\bs{A}Y)+
\mathcal{H}^{(m-1)}_{\boldsymbol{A,B}}(\bs{A}X,Y)\Big), \quad m\geq 2 \ .
\eea
\beq
% \hspace{90mm}   X,Y \in \mathfrak{X}(M) \ ,
\eeq
%Here the notation $\mathcal{H}^{(2)}_{\boldsymbol{A,B}} (X,Y):= \mathcal{H}_ {\boldsymbol{A,B}}(X,Y)$
%is adopted.
\end{definition}
Each of these brackets   is $\mathbb{R}$-homogeneus of degree $m$ in each of its arguments  $\bs{A}$ and $\bs{B}$, as we shall state in Proposition \ref{pr:3.14}; morever, it is symmetric in the interchange of $\bs{A}$ and $\bs{B}$.

\begin{proposition}[\cite{TT2022CMP}] \label{pr:3.14}

Let $M$ be a differentiable manifold and $\boldsymbol{A},\boldsymbol{B}:\mathfrak{X}(M)\rightarrow \mathfrak{X}(M)$  two commuting operators. For any $f,g,h,k \in C^\infty(M)$ , $X,Y \in \mathfrak{X}(M)$ and for each integer $m\geq 2$, we have
\begin{equation} \label{eq:34}
\mathcal{H}^{(m)}_{f \boldsymbol{I}+g\boldsymbol{A},h\boldsymbol{I}+k\boldsymbol{B}} (X,Y)=g^mk^m\mathcal{H}^{(m)}_{\boldsymbol{A},\boldsymbol{B}} (X,Y)  \ .
\end{equation}
\end{proposition}

\vspace{3mm}

\noi We shall analyze now the case $m=2$ to  offer a new interpretation of our results in \cite{TT2022CMP} in terms of the theory of polarization forms.
\noi Notice that the notation $\mathcal{H}( {\boldsymbol{A,B}})(X,Y):= \mathcal{H}^{(2)}_{\boldsymbol{A,B}} (X,Y)$
will be used. Explicitly, we get

\begin{eqnarray} \nn
 &&\mathcal{H}({\boldsymbol{A,B}})(X,Y)=\Big(\bs{AB}+ \bs{BA}\Big) \llbracket\bs{A}, \bs{B}\rrbracket(X,Y)  +
 \llbracket\bs{A}, \bs{B}\rrbracket(\bs{A}X,\bs{B}Y)+ \llbracket\bs{A}, \bs{B}\rrbracket(\bs{B}X,\bs{A}Y)
 \\
\nn &&-\bs{A}\Big( \llbracket\bs{A}, \bs{B}\rrbracket (X,\bs{B}Y) + \llbracket\bs{A}, \bs{B}\rrbracket (\bs{B}X,Y)\big)-\bs{B}\big( \llbracket\bs{A}, \bs{B}\rrbracket(X,\bs{A}Y)+
 \llbracket\bs{A}, \bs{B}\rrbracket(X,Y) (\bs{A}X,Y)\Big).
\end{eqnarray}

\noi The following result clarifies the geometric meaning of the Haantjes  bracket of level 2. 
\begin{lemma}  [\cite{TT2022CMP}] \label{theo:HBT} 
Let $M$ be a  differentiable manifold  and $\bs{A}$, $\bs{B}:\mathfrak{X}(M)\rightarrow \mathfrak{X}(M)$  two commuting $(1,1)$ tensors which can be simultaneously diagonalized in a local chart of $M$.
Then for any $X,Y \in \mathfrak{X}(M)$, the  Haantjes  bracket $\mathcal{H}_{\boldsymbol{A,B}} (X,Y)$ vanishes.
\end{lemma}

 In order to prove Theorem 20 and Corollary 21 of Ref. \cite{TT2022CMP}, we  introduced there the auxiliary Haantjes brackets    $\mathcal{H}_1 (\boldsymbol{A,B})$,  $\mathcal{H}_2 (\boldsymbol{A,B})$ defined by 
\begin{equation}\label{eq:H12}
\begin{split}
&\mathcal{H}_1(\boldsymbol{A,B})(X,Y):=\bs{B}^2\tau_{\bs{A}}(X,Y)+\tau_{\bs{A}}(\bs{B}X,\bs{B}Y) -\bs{B}\big( \tau_{\bs{A}}(\bs{B}X,Y)+\tau_{\bs{A}}(X,\bs{B}Y)\big) \\
&\qquad\qquad+\bs{A}^2\tau_{\bs{B}}(X,Y)+\tau_{\bs{B}}(\bs{A}X,\bs{A}Y) -\bs{A}\big( \tau_{\bs{B}}(\bs{A}X,Y)+\tau_{\bs{B}}(X,\bs{A}Y)\big )\ ,\\
% \end{equation}
% \begin{equation}\label{eq:H2}
% \begin{split}
&  \mathcal{H}_2(\boldsymbol{A,B}) (X,Y):
=(\bs{A}\bs{B}+\bs{B}\bs{A})\tau_{\bs{A}}(X,Y)+
\tau_{\bs{A}}(\bs{A}X,\bs{BY}) +\tau_{\bs{A}}(\bs{B}X,\bs{A}Y) \\ 
&\qquad\qquad-\bs{A}\big(  \tau_{\bs{A}}(\bs{B}X,Y)+\tau_{\bs{A}}(X,\bs{B}Y)\big)
-\bs{B}\big(  \tau_{\bs{A}}(\bs{A}X,Y)+\tau_{\bs{A}}(X,\bs{A}Y\big)\big)\\
&+ \bs{A}^2 \llbracket\bs{A},\bs{B}\rrbracket(X,Y)+
\llbracket\bs{A},\bs{B}\rrbracket(\bs{A}X,\bs{A}Y)
-\bs{A}\big( \llbracket\bs{A},\bs{B}\rrbracket(\bs{A}X,Y)+\llbracket\bs{A},\bs{B}\rrbracket(X,\bs{A}Y)\big) 
\ .
 \end{split}
 \end{equation}

We are now able to give an algebraic interpretation of the brackets $\mathcal{H} (\boldsymbol{A,B})$, $\mathcal{H}_1 (\boldsymbol{A,B})$ and   $\mathcal{H}_2 (\boldsymbol{A,B})$ in terms of the polarization multilinear form of level $2$.
\begin{proposition} The identities 
$$
\mathcal{P}^{(2)} (\bs{A},\bs{A},\bs{B},\bs{B})(X,Y)=4\, \big(\mathcal{H}(\bs{A},\bs{B} )+\mathcal{H}_1(\bs{A},\bs{B})\big)(X,Y)
$$
$$
\mathcal{P}^{(2)} (\bs{A},\bs{A},\bs{A},\bs{B})(X,Y)=6\, \mathcal{H}_2(\bs{A},\bs{B} )(X,Y)
$$
$$
\mathcal{P}^{(2)} (\bs{A},\bs{B},\bs{B},\bs{B})(X,Y)=6\, \mathcal{H}_2(\bs{B},\bs{A} )(X,Y) \ ,
$$
hold.
\end{proposition}
\begin{proof}
These identities are derived directly by Eq. \eqref{eq:Pol2}, setting $\bs{A}_1=\bs{A}, \bs{A}_2=\bs{B},\bs{A}_3=\bs{A},\bs{A}_4=\bs{B}$.
\end{proof}
We point out that in the light of the latter result, Corollary 21 of \cite{TT2022CMP} turns out to be a special case of our Theorem \ref{th:main}.

\end{document}